\documentclass[a4paper,12pt]{article}

\usepackage[russian]{babel}
\usepackage{amsmath, amsthm, amsfonts, amssymb, graphicx}

\theoremstyle{plain}
\newtheorem{theorem}{Theorem}[section]

\newtheorem{lemma}[theorem]{Lemma} 

\theoremstyle{remark}
\newtheorem{remark}{Remark}[section]

\theoremstyle{definition}

\begin{document}

\title{\textbf{On quantum waveguide with shrinking potential}}

\author{A.~Bikmetov, R.~Gadyl'shin}
\date{}
\maketitle

\vspace{-1 true cm}

\begin{quote}
{\small {\em M.~Akmulla Bashkir State Pedagogical University, October Revolution
St.~3a,
\\
450000 Ufa, Russia
\\
E-mail: \texttt{bikmetovar@yandex.ru, gadylshin@yandex.ru}}}
\end{quote}

\begin{abstract}
We consider the spectrum of a Schr\"odinger operator in a
{multi\-di\-men\-sio\-nal} cylinder perturbed by a shrinking potential. We study
the phenomenon of a new eigenvalue emerging from the threshold of the essential
spectrum and give the sufficient conditions for such eigenvalues to emerge. If
such eigenvalues exist, we construct their asymptotic expansions.
\end{abstract}

MCS numbers: 35J10, 35B20, 35P99.

\section*{Introduction}

\medskip

Let $\Omega$ be an open connected domain in $\mathbb{R}^{n-1}$ with an
infinitely smooth boundary if $n\geq 3$ and be a finite interval as $n=2$. We
denote $\Pi=\Omega\times(-\infty, \infty)$. By $-\Delta_{\mathcal{D}}$ we
indicate the Dirichlet Laplacian in $L^2(\Pi)$ with domain $H_0^2(\Pi)$, where
$H_0^2(\Pi)$ is a subset of the functions in $H^2(\Pi)$ vanishing on the
boundary of $\partial\Pi$. It is known that its spectrum
$\sigma(-\Delta_{\mathcal{D}})$ consists only of its essential part $\sigma_e
(-\Delta_{\mathcal{D}})=[\mu_0,\infty)$, where $\mu_0$ is the minimal eigenvalue
of Laplacian $-\Delta^\Omega_{\mathcal{D}}$ in $L^2(\Omega)$ with the domain
$H^2_0(\Omega)$. By the essential spectrum $\sigma_e(\mathcal{A})$ of an
operator $\mathcal{A}$ we mean the set of $\lambda$ so that there exists a
bounded noncompact sequence $u_n\in\mathcal{D}_{\mathcal{A}}$ satisfying the
convergence  $(\mathcal{A}-\lambda I)u_n\to 0$ as $n\to\infty$.

In \cite{DE} they considered the Schr\"odinger operator
\begin{equation*}
   H_h=-\Delta_{\mathcal{D}}+hV,
\end{equation*}
in $L^2(\Pi)$ with the domain $H_0^2(\Pi)$, where $0<h\ll 1$ was a
small parameter, and the potential $V(x)$ was supposed to be
measurable and bounded in $\Pi$, and
\begin{equation*}
\left<|x||V|\phi_0^2\right>\,<\infty.
\end{equation*}
Hereinafter
\begin{equation*} \left<g\right>:=\int\limits_{\Pi}g(x)\mathrm{d}x,
\end{equation*}
and $\phi_0$ is a normalized eigenfunction of
$-\Delta^\Omega_{\mathcal{D}}$ associated with the minimal
eigenvalue $\mu_0$ of $-\Delta^\Omega_{\mathcal{D}}$.  It is known
that  $\phi_0$ can be chosen real-valued. Applying Birman-Schwinger
principle, it was shown that for $h$ small enough the operator $H_h$
has the unique isolated eigenvalue $e(h)$ below $\mu_0$ if and only
if
%
\begin{equation}\label{GR1}
 \left<V\phi^2_0\right>< 0,
\end{equation}
or 
\begin{equation}\label{GR2}
\left<V\phi^2_0\right> = 0.
\end{equation}
If the inequality (\ref{GR1}) or the equality (\ref{GR2}) hold true,
the asymptotic expansion of $e(h)$ was obtained. In particular, if
the inequality (\ref{GR1}) holds, it was shown that
\begin{equation}\label{5}
\begin{aligned}
e(h)=\mu_0-\frac{h^2}{4}\left<V \phi_0^2\right>^2+O(h^3),\quad h\to
0.
\end{aligned}
\end{equation}

The axisymmetric case with a potential depending nonlinearly on a
{pa\-ra\-me\-ter} $h$ was studied in \cite{BCEZ}. It was shown that
the next terms of the expansion of the potential with respect to a
small parameter have an influence on the necessary and sufficient
conditions for an eigenvalue to emerge.

In the present paper we study the perturbation of a quantum
waveguide by a potential which depends on the parameter $h$ as
follows. As $h\to0$, its support shrinks to a point, while the
values of the potential can increases unboundedly. Employing the
results of \cite{GRR1}, \cite{GRR}, we obtain the sufficient
condition for both the presence and absence of an eigenvalue
emerging from the threshold of the essential spectrum. In the former
case we construct the two-terms asymptotics for the emerging
eigenvalues.

\section{Main results}

Without loss of generality we assume that the domain $\Pi$ contains the origin.

We study the perturbed Schr\"odinger operator
\begin{equation*}\label{GR}
   \mathcal{H}^h=-\Delta_{\mathcal{D}}
   +h^{-\alpha} \mathcal{V}_h,
\end{equation*}
in $L^2(\Pi)$ with the domain $H^2_0(\Pi)$, where $\alpha<1$ is a
fixed number,
$\mathcal{V}_h(x)=\mathcal{V}\left(\frac{x}{h}\right)$,
$\mathcal{V}$ is a piecewise continuous bounded in $\mathbb{R}^n$
function with a compact support, which can be complex-valued. By
analogy with \cite{BG} on can show that the operator $\mathcal{H}^h$
is closed and $\sigma_e(\mathcal{H}^h)=[\mu_0,\infty)$.
Since the
function $\mathcal{V}$ is complex-valued, the operator $\mathcal{H}^h$ is
non-self-adjoint.

It follows from the definition of the potential, that its support shrinks to a
point as $h\to 0$, while  the values increasing unboundedly (as $\alpha>0$). It
is clear that it is impossible to reduce by a change of the variables the
eigenvalue problem of $\mathcal{H}^h$ to that of $H_h$.

Denote
\begin{equation*}
\begin{aligned}
&\beta_n(h)=h\sqrt{|\ln h|},\quad &&n=2,
\\
&\beta_n(h)=h,\quad &&n\geq 3.
\end{aligned}
\end{equation*}
The main result of the paper is following theorem.

\begin{theorem}\label{th1}
Let $\alpha<1$. If
\begin{equation*}
\mathrm{Re}\left<\mathcal{V}\right> > 0,
\end{equation*}
then the operator $\mathcal{H}^h$ has no eigenvalues converging to
$\mu_0$.

If
\begin{equation}\label{GR3}
\mathrm{Re}\left<\mathcal{V}\right> < 0,
\end{equation}
then the operator $\mathcal{H}^h$ has a unique, and, in addition, simple
eigenvalue converging to $\mu_0$. Moreover, its asymptotics reads as follows
\begin{equation}\label{g5}
e(h)=\mu_0-\frac{h^{2(n-\alpha)}}{4}
\left(\phi_0^2(0)\left<\mathcal{V}\right>\right)^2\left(1+
\mathcal{O}(h+h^{-\alpha}\beta_n(h))\right).
\end{equation}
\end{theorem}

The formulas (\ref{g5}) is an analogue of the formula (\ref{5}) and
implies that the eigenvalue converging to $\mu_0$ lie outside
$\sigma_e(\mathcal{H}^h)$. For the real-valued function
$\mathcal{V}$ the condition (\ref{GR3}) obviously becomes
%
\begin{equation}\label{GR4}
\left<\mathcal{V}\right><0.
\end{equation}
The conditions (\ref{GR3}) are (\ref{GR4}) are analogues of
(\ref{GR1}).

For the operator $\mathcal{H}^h$ an analogue of critical case
(\ref{GR2}) is the identity
\begin{equation}\label{!+}
\mathrm{Re}\left<\mathcal{V}\right>=0
\end{equation}
(which is
\begin{equation*}\label{!}
\left<\mathcal{V}\right>=0
\end{equation*}
for real-valued potentials $\mathcal{V}$).  However, while the
operator $H_h$ has an eigenvalue converging to $\mu_0$ in the
critical case (\ref{GR2}), it will follow from the proof of
Theorem~\ref{th1} that the critical condition (\ref{!+}) is not
sufficient for the existence of an eigenvalue of $\mathcal{H}^h$
converging to $\mu_0$ as $\alpha<0$.

In the concluding section we show that the same situation occurs for
the strip in the critical case (\ref{!}) as $0\leq
\alpha<\frac{1}{2}$ (cf. Remark~\ref{rmGR2}).

\begin{remark}\label{rmGR1}
In the proof of Theorem~\ref{th1} we employ substantially the
results of \cite{GRR}. Moreover, these results allow us to consider
not only real-valued potentials $\mathcal{V}$, but also
complex-valued ones. This is the reason why we consider
complex-valued potentials $\mathcal{V}$. In particular it means that
the perturbed operator $\mathcal{H}^h$ is not necessarily
self-adjoint.
\end{remark}

\section{Preliminaries}

In \cite{GRR} the operator
 \begin{equation*}
   \mathcal{H}_\varepsilon=-\Delta_{\mathcal{D}}
   +\varepsilon\mathcal{L}_\varepsilon,
\end{equation*}
was considered, where $0<\varepsilon\ll 1$ is a small parameter, and
$\mathcal{L}_\varepsilon$ is an arbitrary localized operator of
second order (not necessarily symmetric). Namely,
$\mathcal{L}_\varepsilon:H^2_{loc}(\Pi)\to L^2(\Pi; Q)$, where $Q$
is a fixed bounded domain lying in $\Pi$, and
\begin{equation*}
L^2(\Pi; Q):=\{u: u\in L^2(\Pi), \mathrm{supp}~u\subseteq
\overline{Q}\}.
\end{equation*}
The operator $\mathcal{L}_\varepsilon$ was assumed to be bounded
uniformly in $\varepsilon$,
\begin{equation}\label{L}
\|\mathcal{L}_\varepsilon u\|_{L^2(\Pi)}\leq C_1\|u\|_{H^2(Q)},
\end{equation}
where $C_1$ is a constant independent of $\varepsilon$. By analogy
with \cite{BG} one can check that the operator
$\mathcal{H}_\varepsilon$ in $L^2(\Pi)$ with the domain $H^2_0(\Pi)$
is closed and $\sigma_e(\mathcal{H}_\varepsilon)=[\mu_0,\infty)$.

For small complex $k$ we define a linear operator
$A(k)\,:\,L^2({\Pi};Q)\to H^2_{loc}({\Pi})$ as follows
\begin{equation}\label{A}
\begin{aligned}
A(k)g:=& \frac{\phi_0(x')}{2k}\int\limits_{\Pi}
\mathrm{e}^{-k|x_n-t_n|}\phi_0(t') g(t)\,\mathrm{d}t\\&+
\sum\limits_{j=1}^\infty \frac{\phi_j(x')}{2K_j(k)}\int\limits_{\Pi}
\mathrm{e}^{-K_j(k)|x_n-t_n|}\phi_j(t') g(t)\,d\mathrm{t},
\end{aligned}
\end{equation}
where $x'=(x_1, x_2,...,x_{n-1})$, $K_j(k)=\sqrt{\mu_j-\mu_0+k^2}$,
and $\mu_j$ and  $\phi_j$ are the eigenvalues of
$-\Delta^\Omega_{\mathcal{D}}$ and the associated eigenfunctions
orthonormalized in $L^2(\Omega)$. We note that
$A(k)=\mathcal{R}_\mathcal{D}(\mu_0-k^2)$ as $\mathrm{Re}\,k>0$,
where $\mathcal{R}_\mathcal{D}(\lambda)$ is the resolvent of the
operator $-\Delta_\mathcal{D}$. We denote by $I$ the identity
mapping, and by $\mathcal{\widetilde{R}}_\mathcal{D}(k):
L^2({\Pi};Q)\to H^2(Q)$ and $ T_\varepsilon(k): L^2({\Pi};Q)\to
L^2({\Pi};Q)$ we indicate the operators introduced as follows
\begin{align}
 \mathcal{\widetilde{R}}_\mathcal{D}(k)g:=
&A(k)g-\frac{1}{2k}
 \left<g\phi_0\right>,\label{A1}\\
  T_\varepsilon(k)g:=
&{\mathcal
L}_\varepsilon\mathcal{\widetilde{R}}_\mathcal{D}(k)g.\label{gr1}
\end{align}

The following results were obtained in \cite{GRR}.

\begin{theorem}\label{GRR}
For $k$ small enough the equation
\begin{equation}\label{9}
2k+\varepsilon\mathcal{F}_\varepsilon(k)=0,
\end{equation}
where
\begin{equation*}
\mathcal{F}_\varepsilon(k)=\left<\phi_0 (I+\varepsilon
T_\varepsilon(k))^{-1}\mathcal{L}_\varepsilon \phi_0\right>,
\end{equation*}
has the unique solution $k_\varepsilon$.

If $\mathrm{Re}\,k_\varepsilon<0$, then the operator $\mathcal{H}_\varepsilon$
has no eigenvalues converging to $\mu_0$.

If $\mathrm{Re}\,k_\varepsilon>0$, then the operator $\mathcal{H}_\varepsilon$
has an eigenvalue converging to $\mu_0$ which is determined by the identity
\begin{equation}\label{GR9}
e_\varepsilon=\mu_0-k^2_\varepsilon.
\end{equation}
\end{theorem}

The definition of the function $\mathcal{F}_\varepsilon(k)$ implies
that for any $N\geq2$ the representation
\begin{equation}\label{gr4}
\mathcal{F}_\varepsilon(k)=\left<\phi_0\mathcal{L}_\varepsilon
\phi_0\right>+\sum\limits_{j=1}^{N-1}(-1)^j\varepsilon^j\left<\phi_0T_\varepsilon^j(k)\mathcal{L}_\varepsilon
\phi_0\right>+\varepsilon^N\mathcal{F}_{\varepsilon,N}(k),
\end{equation}
holds true, where $\mathcal{F}_{\varepsilon,N}(k)$ are holomorphic
in $k$ functions bounded uniformly in $\varepsilon$.

\section{Proof of Theorem~\ref{th1}}

Before proceeding to the proof of Theorem~\ref{th1} we prove an
auxiliary statement. Let $Q$  be a bounded domain with
$C^2$-boundary lying in $\Pi$ and containing the origin. By $h Q$ we
denote $h^{-1}$-fold contraction of a set $Q$.

\begin{lemma}\label{l2.1}
For any function $u\in {H^2}(\Pi)$ the inequality
\begin{gather}\label{2.8}
\int\limits_{hQ}|u|^2 \mathrm{d}x\leq
C_2\beta_n^2(h)\|u\|_{H^2(Q)}^2,
\end{gather}
holds true, where the constant $C_2$ is independent of $h$.
\end{lemma}

\begin{proof}
Let $B$ be a bounded set in $\mathbb{R}^n$, $n\geq 2$, containing
the origin. By $H^m_0(B)$ we denote the set of the functions in
$H^m(B)$ vanishing on $\partial B$.

In (\cite[Ch. 3, \S 5, Lm. 5.1]{OIoSh}) and \cite{OPY} for
$n\geqslant3$ and $n=2$, respectively, it was shown that for any
function $U\in{H^1_0}(B)$ the inequality
\begin{equation}\label{2.4}
\int\limits_{h Q}|U|^2\mathrm{d}x\leq C_3\beta_n^2(h)\int\limits_{B}
|\nabla U|^2 \mathrm{d}x
\end{equation}
holds true, where $C_3$ is a constant depending on the domain $B$.
This inequality is the corollary the Hardy inequality (see
\cite{KonOl}). We choose $B$ such that $\overline{Q}\subset B$.

It is well-known (see, for instance, \cite[Ch. 3, \S 4, Thm.
1]{Mikh}) that for any function $u\in {H^2}(Q)$ there exists a
continuation $U\in H^2_0(B)$ such that
\begin{equation}\label{qwerty}
\|U\|_{H^2(B)}\leq C_4 \|u\|_{H^2(Q)},
\end{equation}
where the constant $C_4$ depends only on $B$ and $Q$.

By (\ref{2.4}), (\ref{qwerty}) for any function $u\in H^2(\Pi)$ we
derive the inequalities
\begin{equation*}
\begin{aligned}
\int\limits_{h Q}|u|^2\mathrm{d}x=&\int\limits_{h
Q}|U|^2\mathrm{d}x\leq C_3\beta_n^2(h)\int\limits_{B}|\nabla U|^2
\mathrm{d}x \leq C_5\beta_n^2(h)\|U\|^2_{H^2(B)}\\ \leq&
C_4^2C_5\beta_n^2(h)\|u\|^2_{H^2(Q)}.
\end{aligned}
\end{equation*}
The proof is complete.
\end{proof}

We proceed to the proof of Theorem~\ref{th1}. We denote
\begin{equation}\label{gr}
\varepsilon(h):=h^{-\alpha}\beta_n(h).
\end{equation}
The definition of $\beta_n(h)$ yields that $\varepsilon(h)>0$ and
\begin{equation*}
\varepsilon(h)\underset{h\to0}\to0.
\end{equation*}
It is also obvious that for $h>0$ small enough there exists a function
$h(\varepsilon)$ inverse to the function $\varepsilon(h)$ such that
\begin{equation*}
\begin{aligned}
&h(\varepsilon)>0,\qquad
h(\varepsilon)\underset{\varepsilon\to0}\to0.
\end{aligned}
\end{equation*}

In view of the definition (\ref{gr}) the perturbing potential of the
operator $\mathcal{H}^h$ can be represented as
\begin{equation*}
h^{-\alpha}\mathcal{V}_h=\varepsilon(h) \frac{\mathcal{V}_h}{
    \beta_n(h)}.
\end{equation*}

Let us show that the operator $\mathcal{L}_\varepsilon$ defined as
the multiplication by the function
$\beta_n^{-1}(h(\varepsilon))\mathcal{V}_{h(\varepsilon)}$
\begin{equation}\label{oper}
  \mathcal{L}_\varepsilon u:=\frac{\mathcal{V}_{h(\varepsilon)}u}{\beta_n(h(\varepsilon))},
\end{equation}
satisfies the estimate (\ref{L}), if we treat it as the operator
from $H^2(Q)$ to $L^2(\Pi; Q)$. In other words,
\begin{equation}\label{p3}
\left\|\frac{\mathcal{V}_h}{\beta_n(h)}u\right\|_{L^2(\Pi)}\leq
C_6\|u\|_{H^2(Q)},
\end{equation}
where $C_6$ is a constant independent of $h$.

Indeed, without loss of generality we can assume that the support of the
function $\mathcal{V}$ lies in a bounded domain $Q\subset\Pi$. Then by
Lemma~\ref{l2.1}, for each $u\in H^2(\Pi)$ we consequently obtain
\begin{equation*}
\begin{aligned}
\int\limits_\Pi\left| \mathcal{V}_h(x)u(x)\right|^2
\mathrm{d}x=&\int\limits_{h Q}\left|
\mathcal{V}\left(\frac{x}{h}\right)u(x)\right|^2 \mathrm{d}x\leq
\max_{x\in \overline{Q}}|\mathcal{V}(x)| \int\limits_{h
Q}\left|u(x)\right|^2 \mathrm{d}x
\\
\leq&C_7\beta_n(h)\|u\|^2_{H^2(Q)}.
\end{aligned}
\end{equation*}
This inequality implies the estimate (\ref{p3}), and therefore the
estimate (\ref{L}) for the operator defined by the identity
(\ref{oper}). Hence, for the operator defined by (\ref{oper})
Theorem~\ref{GRR} holds true, where the operator $T_\varepsilon(k)$
in (\ref{gr1}) is introduced as
\begin{equation*}
T_{\varepsilon(h)}(k)g=\beta_n^{-1}(h) {\mathcal
V}_h\mathcal{\widetilde{R}}_\mathcal{D}(k)g.
\end{equation*}
Together with (\ref{oper}) and (\ref{gr}) it follows that for any
natural $N$ the identity (\ref{gr4}) becomes
\begin{equation}\label{gr4+}
\begin{aligned}
\mathcal{F}_\varepsilon(k)=&\beta_n^{-1}(h)\left(\left<\phi_0^2{\mathcal
V}_h\right>+\sum\limits_{j=1}^{N-1}(-1)^j
h^{-j\alpha}\left<\phi_0({\mathcal
V}_h\mathcal{\widetilde{R}}_\mathcal{D}(k))^j({\mathcal V}_h
\phi_0)\right>\right)\\&+(h^{-\alpha}\beta_n(h))^N\mathcal{F}_{\varepsilon,N}(k).
\end{aligned}
\end{equation}

Denote
\begin{equation*}
\begin{aligned}
&\Phi_0(\xi')=\sum\limits_{q=1}^{n-1} \frac{\partial
\phi_0}{\partial \xi_q}(0)\xi_q.
    \end{aligned}
\end{equation*}
By direct calculations we check that
\begin{align}
&
\begin{aligned}
\left<\phi_0^2\mathcal{V}_h\right>=&\int\limits_{h Q}
    \mathcal{V}\left(\frac{x}{h}\right)\phi_0^2(x')\mathrm{d}x=
h^{n}\int\limits_{Q}\mathcal{V}(\xi)\phi^2_0(h\xi')\mathrm{d}\xi\\=
    &
    h^n\phi_0^2(0)\left<\mathcal{V}\right>+h^{n+1}2\phi_0(0)\left<\Phi_0\mathcal{V}\right>
    +\mathcal{O}(h^{n+2}),\label{pos}
\end{aligned}
\\
&\left\|\phi_0
   \mathcal{V}_h\right\|^2_{L^2(\Pi)}=
h^{n}\int\limits_{Q}\mathcal{V}^2(\xi)\phi^2_0(h\xi')\mathrm{d}\xi=\mathcal{O}(h^{n}).\label{pos+}
\end{align}
Lemma~\ref{l2.1} and the definition of the operator
$\mathcal{\widetilde{R}}_\mathcal{D}(k)$ imply that for each
function $g\in L^2(\Pi; Q)$ the estimate
\begin{equation}\label{60}
    \|\mathcal{\widetilde{R}}_\mathcal{D}(k)g\|_{L^2(hQ)}\leq
    C_2\beta_n(h)\|\mathcal{\widetilde{R}}_\mathcal{D}(k)g\|_{H^2(Q)}
    \leq C_2C_{\mathcal{R}}\beta_n(h)\|g\|_{L^2(\Pi)}
\end{equation}
holds true.

We denote
\begin{equation*}
a_j=\left|\left<\phi_0({\mathcal
V}_h\mathcal{\widetilde{R}}_\mathcal{D}(k))^j({\mathcal V}_h
\phi_0)\right>\right|,\qquad C_\mathcal{V}=\underset{x\in
\overline{G}}{\max}|\mathcal{V}(x)|.
\end{equation*}
Then by (\ref{60}) and (\ref{pos+}) for $j\geq1$ we
have
\begin{equation}\label{51}
\begin{aligned}
a_j\leq& \left\|\phi_0
   \mathcal{V}_h\right\|_{L^2(hQ)}\left\|\mathcal{\widetilde{R}}_\mathcal{D}(k)(\mathcal{V}_h
   {\widetilde{R}}_\mathcal{D}(k))^{j-1}
   (\mathcal{V}_h\phi_0)\right\|_{L^2(hQ)}.
   \\
\leq &C_2 C_{\mathcal{R}}\beta_n(h)\left\|\phi_0
   \mathcal{V}_h\right\|_{L^2(\Pi)}\left\|(\mathcal{V}_h
   {\widetilde{R}}_\mathcal{D}(k))^{j-1}
   (\mathcal{V}_h\phi_0)\right\|_{L^2(\Pi)}\\
= &C_2 C_{\mathcal{R}}\beta_n(h)\left\|\phi_0
   \mathcal{V}_h\right\|_{L^2(\Pi)}\left\|(\mathcal{V}_h
   {\widetilde{R}}_\mathcal{D}(k))^{j-1}
   (\mathcal{V}_h\phi_0)\right\|_{L^2(hQ)}\\
\leq &C_2 C_{\mathcal{R}}C_\mathcal{V}\beta_n(h)\left\|\phi_0
   \mathcal{V}_h\right\|_{L^2(\Pi)}\left\|{\widetilde{R}}_\mathcal{D}(k)(\mathcal{V}_h
   {\widetilde{R}}_\mathcal{D}(k))^{j-2}
   (\mathcal{V}_h\phi_0)\right\|_{L^2(hQ)}\\
\leq &\cdots\\
\leq &\left(C_2
C_{\mathcal{R}}C_\mathcal{V}\beta_n(h)\right)^{j-1}\left\|\phi_0
   \mathcal{V}_h\right\|_{L^2(\Pi)}\left\|
   {\widetilde{R}}_\mathcal{D}(k)
   (\mathcal{V}_h\phi_0)\right\|_{L^2(hQ)}
  .
\end{aligned}
\end{equation}

Together with (\ref{pos+}) and (\ref{60}) it implies
that
\begin{equation}\label{B12}
a_j =\mathcal{O}(h^{n}\beta_n^j(h)).
\end{equation}

The number $N$ in the estimate (\ref{gr4+}) is arbitrary that
together with (\ref{B12}), (\ref{pos}) and (\ref{gr}) yields
\begin{equation*}\label{gr5}\begin{aligned}
\varepsilon\mathcal{F}_\varepsilon(k)=&h^{n-\alpha}\left(\phi_0^2(0)
\left<\mathcal{V}\right>+h2\phi_0(0)\left<\Phi_0\mathcal{V}\right>+
\mathcal{O}\left(h^2+h^{-\alpha}\beta_n(h)\right)\right).
\end{aligned}
\end{equation*}
It follows by (\ref{9}) that
\begin{equation}\label{GR7}
k_{\varepsilon(h)}=\frac{1}{2}h^{n-\alpha}\left(\phi_0^2(0)
\left<\mathcal{V}\right>+h2\phi_0(0)\left<\Phi_0\mathcal{V}\right>
+\mathcal{O}\left(h^2+h^{-\alpha}\beta_n(h)\right)\right).
\end{equation}
And, finally, the last identity and Theorem~\ref{GRR} lead us to
Theorem~\ref{th1}.

\begin{remark}\label{rmGR2} Let us consider the critical case
(\ref{!+}) for $\alpha<0$. It follows from the identity (\ref{GR7})
that if
\begin{equation*}
\phi_0(0)\,\mathrm{Re}\left<\Phi_0\mathcal{V}\right> > 0,
\end{equation*}
then $k_{\varepsilon(h)}<0$, and therefore by Theorem~\ref{GRR} the
operator $\mathcal{H}^h$ has no eigenvalues converging to $\mu_0$.

If
\begin{equation}\label{GR10}
\phi_0(0)\,\mathrm{Re}\left<\Phi_0\mathcal{V}\right> < 0,
\end{equation}
then the identity (\ref{GR7}) implies that $k_{\varepsilon(h)}>0$
and thus by Theorem~\ref{GRR}  the operator $\mathcal{H}^h$ has a
unique, and, in addition, simple eigenvalue converging to $\mu_0$.
For the real-valued functions $\mathcal{V}$ the inequality
(\ref{GR10}) casts into the form
\begin{equation*}
\phi_0(0)\left<\Phi_0\mathcal{V}\right> < 0,
\end{equation*}
and by the identities (\ref{GR7}) and (\ref{GR9}) the asymptotics of
the eigenvalue reads as follows
\begin{equation*}\label{g5B}
e(h)=\mu_0-h^{2(n+1-\alpha)}
\left(\phi_0(0)\left<\Phi_0\mathcal{V}\right>\right)^2\left(1+
\mathcal{O}(h+h^{-1-\alpha}\beta_n(h))\right).
\end{equation*}
\end{remark}

%

\section*{Critical case in the strip as $0\leq\alpha<\frac{1}{2}$}

Let us consider the case $n=2$,
\begin{equation*}\label{!1}
\mathcal{V}(t)=v(t_1)\widetilde{v}(t_2),
\end{equation*}
where
\begin{equation*}\label{!2}
\widetilde{v}(t_2)=
\begin{cases}
1&\text{as $|t_2|<1$}\\
0&\text{as $|t_2|>1$}.
\end{cases}
\end{equation*}
In this case the identity (\ref{pos}) becomes
\begin{equation}\label{posD}
\begin{aligned}
   \left<\phi_0^2\mathcal{V}_h\right>=&2h^2
   \left(\phi_0^2(0)\left<v\right>'+2h\phi_0(0)\phi'_0(0)
   \int\limits_\Omega
   v(t_1)t_1\mathrm{d}t_1\right)+
   \mathcal{O}(h^4),
\end{aligned}
\end{equation}
where
\begin{equation*}
\left<g\right>':=\int\limits_\Omega
   g(t_1)\mathrm{d}t_1,
\end{equation*}
and by (\ref{A}) and (\ref{A1}) the function
${\widetilde{R}}_\mathcal{D}(k) (\mathcal{V}_h\phi_0)$ reads as
follows
\begin{equation}\label{A2}
\begin{aligned}
{\widetilde{R}}_\mathcal{D}(k)
   (\mathcal{V}_h\phi_0)=&b_0(x_2;k)\phi_0(x_1)\left<v_h\phi^2_0\right>'
\\&+ \sum\limits_{j=1}^\infty
b_j(x_2;k)\frac{\phi_j(x_1)}{K^2_j(k)}\left<\phi_j
v_h\phi_0\right>',
\end{aligned}
\end{equation}
where $v_h(x_1)=v(x_1h^{-1})$,
\begin{equation}\label{A3}
\begin{aligned}
b_0(x_2;k)=&\frac{1}{k}\left(\frac{1}{k}\left(1-e^{-kh}\cosh(kx_2)\right)-h\right),\\
b_j(x_2;k)=&1-e^{-K_j(k)h}\cosh(K_j(k)x_2),\quad j\geq1.
\end{aligned}
\end{equation}

In the case considered $Q=\Omega\times(-1,1)$, $\Omega$ is an
interval $(\omega_-,\omega_+)$, $\pm\omega_\pm>0$. Denote
$Q_h:=\Omega\times(-h,h)$. It follows from (\ref{A2}) and (\ref{A3})
that
\begin{equation}\label{A4}
\|\mathcal{\widetilde{R}}_\mathcal{D}(k)(\mathcal{V}_h\phi_0)\|_{L^2(Q_h)}^2\leq
 Ch \sum\limits_{j=0}^\infty
\frac{|\left<\phi_j v_h\phi_0\right>'|^2}{\mu_j}.
\end{equation}
It is well-known that for the solution to the boundary value problem
\begin{equation}\label{A5}
-\frac{\mathrm{d}^2U}{\mathrm{d}x_1^2}=v_h\phi_0,\quad
x_1\in\Omega,\qquad U(\omega_\pm)=0
\end{equation}
the identity
\begin{equation*}
    \|U\|_{L^2(\Omega)}^2=\sum\limits_{j=0}^\infty
\frac{|\left<\phi_j v_h\phi_0\right>'|^2}{\mu_j}
\end{equation*}
holds true. Hence, by (\ref{A4}) and the inclusion $hQ\subset Q_h$
\begin{equation}\label{A6}
 \|\mathcal{\widetilde{R}}_\mathcal{D}(k)(\mathcal{V}_h\phi_0)\|_{L^2(hQ)}^2\leq
 Ch \|U\|_{L^2(\Omega)}^2.
\end{equation}
We represent the solution to (\ref{A5}) as
\begin{equation*}
U=U_0+U_1+\widetilde{U},
\end{equation*}
where $U_j$ and $\widetilde{U}$ solve the problems
\begin{align}\label{A7}
-\frac{\mathrm{d}^2U_0}{\mathrm{d}x_1^2}=&v_h\phi_0(0),\quad
x_1\in\Omega,\qquad U_0(\omega_\pm)=0,\\
-\frac{\mathrm{d}^2U_1}{\mathrm{d}x_1^2}=&v_h\phi'_0(0) x_1,\quad
x_1\in\Omega,\qquad U_1(\omega_\pm)=0,\label{A8}\\
-\frac{\mathrm{d}^2\widetilde{U}}{\mathrm{d}x_1^2}=&\widetilde{f},\quad
x_1\in\Omega,\qquad
\widetilde{U}(\omega_\pm)=0,\nonumber\\
\widetilde{f}(x_1)=&v_h(x_1)\left(\phi_0(x_1)-\phi_0(0)-\phi'_0(0)
x_1\right)\nonumber.
\end{align}

By direct calculations we check that
\begin{equation*}
\|\widetilde{f}\|^2_{L^2(\Omega)}= \mathcal{O}(h^{5}).
\end{equation*}
Therefore,
\begin{equation}\label{A9}
\|\widetilde{U}\|^2_{L^2(\Omega)}= \mathcal{O}(h^{5}).
\end{equation}

The solution to (\ref{A8}) can be found explicitly,
\begin{equation*}
U_1(x_1)=-h^3\phi'_0(0)\int\limits_{-\infty}^{\frac{x_1}{h}}
\int\limits_{-\infty}^{\xi}v(\eta)\eta\,\mathrm{d}\eta
\,\mathrm{d}\xi+h^2((c_1+c_1'h)x_1+(d_1+d_1'h)),
\end{equation*}
where $c_1$, $c'_1$, $d_1$, and $d'_1$ are explicitly calculated
constants. Hence,
\begin{equation}\label{A10}
\|U_1\|^2_{L^2(\Omega)}= \mathcal{O}(h^{4}).
\end{equation}

We choose $v$ so that
\begin{equation}\label{!!}
\left<v\right>'=0.
\end{equation}
In this case the identity (\ref{!}) holds true and also the solution
to the boundary value problem (\ref{A7}) is as
follows
\begin{equation*}
U_0(x_1)=-h^2\phi_0(0)\int\limits_{-\infty}^{\frac{x_1}{h}}\int\limits_{-\infty}^{\xi}v(\eta)\,\mathrm{d}\eta
\,\mathrm{d}\xi+h^2(c_0x_1+d_0),
\end{equation*}
where $c_0$, and  $d_0$ can be also found explicitly. Therefore,
\begin{equation}\label{A11}
\|U_0\|^2_{L^2(\Omega)}= \mathcal{O}(h^{4}).
\end{equation}

It follows from (\ref{A9})--(\ref{A11}) and (\ref{A6}) that
\begin{equation*}
    \|\mathcal{\widetilde{R}}_\mathcal{D}(k)(\mathcal{V}_h\phi_0)\|_{L^2(hQ)}^2=\mathcal{O}(h^5).
\end{equation*}
This identity, (\ref{pos+}), and (\ref{51}) imply
\begin{equation}\label{A12}
a_j =\mathcal{O}\left(\beta_2^{j-1}(h)h^{\frac{7}{2}}\right).
\end{equation}

The number $N$ in the estimate (\ref{gr4+}) is arbitrary that
together with (\ref{A12}), (\ref{posD}), (\ref{!}), (\ref{gr}) and
(\ref{!!}) yields
\begin{equation*}
\begin{aligned}
\varepsilon\mathcal{F}_\varepsilon(k)=&h^{3-\alpha}\left(4\phi_0(0)
\phi'_0(0)
   \int\limits_\Omega
   v(t_1)t_1\mathrm{d}t_1+\mathcal{O}\left(h^{\frac{1}{2}-\alpha}\right)\right)
\end{aligned}
\end{equation*}
for $n=2$.
It follows by (\ref{9}) that
\begin{equation*}
k_{\varepsilon(h)}=h^{3-\alpha}\left(2\phi_0(0) \phi'_0(0)
   \int\limits_\Omega
   v(t_1)t_1\mathrm{d}t_1+\mathcal{O}\left(h^{\frac{1}{2}-\alpha}\right)\right).
\end{equation*}

Let $|\omega_-|\not=\omega_+$. Then
\begin{equation*}
 \phi_0(0) \phi'_0(0)\not=0.
 \end{equation*}
Obviously, there exists a compactly supported function $\psi(x_1)$
such that
\begin{equation*}
\left<\psi\right>'=0,\qquad \int\limits_\Omega
\psi(t_1)t_1\mathrm{d}t_1\not=0.
\end{equation*}
Then letting $v=\psi$ and $v=-\psi$ we obtain that the quantity
$k_{\varepsilon(h)}$ has different signs in these cases. Then by
Theorem~\ref{GRR} in one case the eigenvalue exists, while in the
other does not.

\section*{Acknowledgments}

 We thank D.~Borisov
 for useful remarks.

The research was supported by Russian Fund of Basic Research under
the contract 08-01-97016-r\_povolzhie, by the grant of the President
of Russia for young scientist and their supervisors, and by the
grant of the President of Russia for leading scientific schools
(NSh-2215.2008.1).

\end{document}